\documentclass[nofootinbib,%
 reprint,
 amsmath,amssymb,
 aps,
]{revtex4-1}
\usepackage{epsfig}

\newenvironment{proof}{ \textbf{Proof:} }{ \hfill $\Box$}
\newtheorem{thm}{Theorem}[section]

\newtheorem{corollary}[thm]{Corollary}

\newtheorem{prop}[thm]{Proposition}

\newtheorem{rem}[thm]{Remark}

\def\bb0{{\mathbb{0}}}
\def\b1{{\mathbf{1}}}

\def\bb{{\mathbf{b}}}

\def\b0{{\mathbf{0}}}

\def\bbE{{\mathbb{E}}}

\def\bbP{{\mathbb{P}}}

\def\bbR{{\mathbb{R}}}

\def\sfF{\mathsf{F}}

\def\sf0{{\mathsf{0}}}

\def\1{{\bf 1}}

\def\ze{\zeta}
\def\ga{\gamma}

\newcommand{\sg}{\sigma}

\def\th{\theta}

\def\sfF{\mathsf{F}}

\def\f{\frac}
\def\tw{\tilde{w}}
\def\tm{\tilde{\mu}}
\begin{document}

\begin{titlepage}
\begin{center}
{\bf Autoregressive Cascades on Random Networks} \\
\vspace{0.2in} {Srikanth K. Iyer\footnote{skiyer@math.iisc.ernet.in, research supported in part by UGC center for advanced studies.}}\\
Department of Mathematics,
Indian Institute of Science, Bangalore, India. \\
\vspace{0.1in}
Rahul Vaze \footnote{vaze@tcs.tifr.res.in, research supported in part by INSA young scientist award grant.}\\
School of Technology and Computer Science, Tata Institute of Fundamental Research,
Homi Bhabha Road, Mumbai, India.\\
\vspace{0.1in}
Dheeraj Narsimha \footnote{dheeraj.narasimha@gmail.com.} \\
Birla Institute of Technology and Science (Goa)\\
Goa, India.
\end{center}
\vspace{0.05in}
%
%

\begin{center} {\bf Abstract} \end{center}


{This paper considers a model for cascades on random networks in
which the cascade propagation at any node depends on the load at
the failed neighbor, the degree of the neighbor as well as the load at
that node. Each node in the network bears an initial load that is
below the capacity of the node. The trigger for the cascade
emanates at a single node or a small fraction of the nodes from
some external shock. Upon failure, the load at the failed node
gets divided randomly and added to the existing load at those
neighboring nodes that have not yet failed. Subsequently, a
neighboring node fails if its accumulated load exceeds its
capacity. The failed node then plays no further part in the
process. The cascade process stops as soon as the accumulated load at
all nodes that have not yet failed is below their respective capacities.
The model is shown to operate in two regimes, one in which the
cascade terminates with only a finite number of node failures. In
the other regime there is a positive probability that the cascade
continues indefinitely. Bounds are
obtained on the critical parameter where the phase transition occurs.} \\

%
\begin{center}
\today
\end{center}
%
\end{titlepage}

\section{Introduction}
\label{sec:intro}  A model for the evolution of cascades in random
networks was introduced in \cite{watts2002simple}, to study the
spread of ideas, opinion, technology etc. In this model, the nodes
of a network are regarded as agents and the interactions between
agents are modeled as links in the network.  An agent initially in
state 0 will adopt a new idea (state 1) as soon as a fraction of
its neighbors who have adopted the new idea exceeds a threshold.
The existence of a phase transition (as a function of the
threshold) was demonstrated and a condition for the existence of
global cascades was derived. This model was generalized in
\cite{gleeson2013systemic}, where in addition to the above, the
links are endowed with random weights, whose distribution may
depend on the degree of the two nodes that the link connects. The
motivation for this model comes from the study of systemic risks
in financial networks. A node adopts a new idea provided the sum
of the weights of links to neighboring nodes that have adopted the
idea exceeds a threshold that depends on the degree at that node.
A further generalization is the threshold model
\cite{kempe2003maximizing, venkatramanan2011information}, where an
agent adopts the new idea if the number of its neighbors or the
combined edge weights of the neighbors that have adopted the new
idea is larger than a random threshold.

The model that we study in this paper is motivated by networks
such as the electrical networks or other networks where the nodes
can be thought of as service providers and load sharing occurs via
the network structure. Consider a large random network, where each
node has degree that is distributed according to a specified
distribution $\{w_k, k = 0,1, \ldots\}$ (see
\cite{newman2001random}). Such a network is locally ``tree-like''
when the network size is large, that is, there are few short
cycles. Melnik {\it et. al.} \cite{melnik2011unreasonable} examine
the effectiveness of tree-like networks to study networks with
clustering.

Initially, node $u$ bears a load $L^{(0)}_u$, and has capacity $c
> 0$, where the random variables $L^{(0)}_u$ are independent and
whose distributions are the same as those of the non-negative
random variable $L$. To begin with, let us assume that the support
of $L$ is contained in $[0,c)$ (abbreviated as $L < c$). This
restriction ensures that initially all nodes are stable or active.
At time $t=0$, an external event or a surge occurs at a node that
we label $r$, resulting in the load $L_r$ at that node to exceed
$c$. This causes node $r$ to fail. If the number of neighbors of
$r$, $N(r)=0$, then the cascade stops. If $N(r)=k > 0$, then
suppose that $u_1, \ldots, u_{k}$ are neighbors of $r$ with loads
$L^{(0)}_{u_1}, \ldots L^{(0)}_{u_{k}}$, respectively at time
$t=0$. Then at time $t=1$, a random fraction of the total load of
the failed node $r$ is pushed over to each of the active
neighbors, that is, the load at $u_i$ increases to
\[ L^{(1)}_{u_i} = p_{ru_i} L_r + L^{(0)}_{u_i}, \]
where $\{p_{ru_i} \geq 0, i=1,\ldots, k \}, \sum_{i=1}^{k}
p_{ru_i} =1,$ is random probability mass function.
%
%
Consequently, a node whose new load exceeds its capacity, fails.
Each failed node distributes its entire load among all its active
(non-failed) neighbors randomly as described above. The cascade of
failed nodes terminates if at some time, the resulting loads at
all nodes that are neighbors of nodes that failed in the previous
time step do not exceed their respective capacities. We will refer
to this model as the autoregressive cascade (ARC) model or
process, since the increase in load at a node as a result of
failure of a neighboring node resembles an autoregressive process.
We show the existence of a phase transition, that is, existence of
regions where the cascade dies out in finite time with probability
one, and regions where the process survives indefinitely with
positive probability for infinite networks (or has a giant
component of failed nodes for large finite networks).
%

ARC model is well-suited for studying the behavior of outages in
electrical power networks, where typically a single node failure
can lead to catastrophically wide-spread outages
\cite{Santhi2010}.  The ARC model is linked to the node based
failure model in electrical networks, where nodes fail when the
demand (load) overshoots the supply (capacity), and demand at the
failed node is transferred to its active neighbors. In
\cite{Santhi2010}, simulation analysis is presented for a fully
connected network with exponentially distributed loads. Simpler
node failure models have been studied in
\cite{dobson2004branching, Dobson2005}. In particular, in
\cite{Dobson2005}, with each node failure, the load at every other
active node is uniformly increased by a constant load, while in
\cite{dobson2004branching} each failing node results in the
failure of a random number of nodes sampled at random from a given
distribution.
Both these models ignore the topology of the graph.

There is also a line failure model used in electrical networks
\cite{farina2008probabilistic, xiao2011cascading}, where
transmission lines fail when either the current or voltage exceeds
line's rating, and failure of one transmission line changes the
current/voltages change on other lines and failures propagate
accordingly. ARC model does not capture the line failure model.

Note that the ARC model described above is different from the
usual epidemic models (see for example \cite{watts2002simple}),
where nodes fail or are infected based on some probabilistic or
deterministic mechanism that depends only on the number of
failed/infected neighbors but not the severity of the infection.
These type of models are reviewed in \cite{newman2010networks,
lelarge2012diffusion}  and references therein. In contrast, in the
ARC model, failure is governed by a mechanism of load transfer
whose effect can persist over several generations.

The closest interaction model to the ARC model is the sandpile
model \cite{bak1988self,majumdar1992equivalence,dhar1990abelian},
where at each time slot, a particle is added at a randomly chosen
node. Each node has a fixed capacity, and if the number of
particles at any node exceeds its capacity, the node is said to
topple and all the particles at that node are transferred equally
to all its neighbors. Thus, the newly added particles at a node
due to the toppling of a neighboring node is a fixed deterministic
quantity, and given that a node has toppled, the future toppling
of its neighbors are independent. With the ARC model, however, the
transferred load is a random quantity, that correlates subsequent
node failures. In case of the sandpile model, a toppled node is
allowed to participate in further interactions, while in ARC model
once a node fails it takes no further part in the process. There
is also an inhibition sandpile model \cite{manna1997sandpile},
where a toppled node is stopped temporarily or permanently  from
taking part in further interaction, however, only empirical
results are available for the same. Another similar form of interaction is
the bootstrap percolation, where the failure of a node depends on
the failure of a fixed number of neighbors and not the weights at
the failed neighboring nodes \cite{Lebowitz1989}.

In this paper, we derive sufficient conditions for the cascade in
the ARC process to survive indefinitely with positive probability
or terminate with probability one in finite number of steps. We
couple the ARC process with a suitable Galton-Watson branching
process that lower bounds the growth of failed nodes in the ARC
process. The condition for survival then follows from the
condition for super-criticality of the coupled Galton-Watson
process. For the converse result, we couple the ARC process with a
suitable branching random walk (BRW) where particles are killed
upon breaching a fixed barrier. The bound then follows from a
result in \cite{Biggins1977}  on the finite time termination of
the BRW.


\section{ARC Process}
\label{sec:trees}
%
We first study the evolution of the ARC process on an infinite
graph $G$ with specified degree distribution $\{w_k, k \geq 0\}$.
Although random graphs are only an abstraction, they have been
widely used as first approximations \cite{watts2002simple}. For a
random graph on $n$ vertices with specified degree distribution,
the graph is locally ``tree-like'' in that there are no short
cycles for large $n$. This becomes exact for the infinite graph.

If we start from a single failed node, then the number of
neighbors of this node has distribution $\{w_k\}$. As we go forth,
exploring the component starting from a single vertex or the
progress of the cascade, we are interested in the unexplored nodes
or nodes that have not yet failed. The first and subsequent
generations will then have a size-biased degree distribution (see
\cite{durrett2007random}, Chapter 3, pp71). A first generation
vertex with degree $k$ is $k$ times as likely to be chosen as the
one with degree $1$. So the distribution of degree minus one of
nodes in the first and subsequent generations is given by
\begin{equation}
\tw_{k-1} = \f{k w_k}{\mu}, \qquad k \geq 1, \label{e0}
\end{equation}
where
\begin{equation}
\mu = \sum_{k=1}^{\infty} k w_k. \label{e0a}
\end{equation}
In (\ref{e0}), the subscript $k-1$ is used since we have used one
edge in linking to that vertex.  It is this size-biased
distribution $\tw_k$ given by (\ref{e0}) that is relevant as far
as the long term behavior of the process is concerned.

Since we are interested in the survival or extinction of the ARC
process, we require the underlying graph $G$ to have an infinite
connected component. Thus, without loss of generality we will
assume that $\tm = \sum_{k} k \tw_k > 1$ (see Theorem 3.1.3,
\cite{durrett2007random}) for which there is a positive
probability that the component of $G$ containing any node is
infinite. This condition also ensures that for a finite graph,
there is a giant component containing a large fraction of the
nodes with probability approaching one asymptotically in the size
of the graph.

Each node $u \in G$ carries an initial load $L_u^{(0)}$ and has a fixed
capacity $c$. The random variables $L_u^{(0)}$ are assumed to be
independent and identically distributed. We denote by $L$ a
generic random variable having the same distribution as the
initial loads. Let $\sfF$ be the distribution function of $L$.
Suppose that $\sfF$ is supported on $(0,c).$ Thus, to begin with,
all nodes have loads less than their capacity. This condition is
realistic for power networks where node failures are rare.

At time $t=0$, an external event happens resulting in the failure
of the root node $r$, making $L_r^{(0)} > c$. If $r$ is isolated,
then the cascade stops. Else, the load $L_r^{(0)}$ at $r$ is then
transferred to its neighbors as described below. Given the degree
$N(r)=k$, $k > 0$, of the root node $r$, let $u_1, u_2, \ldots
,u_k$ be the neighboring nodes of the root. Let $\{p_{ru_i}, i=1,
\ldots ,k \}$ be an exchangeable set of non-negative random
variables satisfying $\sum_{i=1}^{k} p_{ru_i} = 1.$ That is, for
any permutation $\sg = (\sg(1), \sg(2), \ldots \sg(k))$ of $\{1,2,
\ldots ,k\}$, we have
\[(p_{ru_1}, p_{ru_2}, \ldots , p_{ru_k}) \stackrel{d}{=}
(p_{ru_{\sg(1)}}, p_{ru_{\sg(2)}}, \ldots , p_{ru_{\sg(k)}}).\]
%
%
%
In particular, the distribution of $p_{ru_i}$ is independent of
$i$, and hence we let $p=p(k)$ to denote the random variable
satisfying
\begin{equation}\label{eq:distp}
p(k) \stackrel{d}{=} p_{ru_i}(k),
\end{equation}
where $\stackrel{d}{=}$ denotes equality in distribution. We will
often write $p_{uv}$ for $p_{uv}(k)$. For each of the $N(r)=k$
nodes, $u_1, \ldots , u_{k}$, attached to the root, the load
$L_{u_i}^{(1)}$ at time $t=1$ becomes
\[ L_{u_i}^{(1)} = p_{ru_i} L_r^{(0)} + L_{u_i}^{(0)}. \]
If $L_{u_i}^{(1)} < c, \ \forall \ i= 1,2, \ldots ,k$, then the
cascade of node failures stops at time $t=1$. Otherwise, node
$u_i$ fails at time $1$ if the load $L_{u_i}^{(1)} \geq c$. If $r$
is the only node to which $u_i$ is connected, then the cascade
goes no further along this path. Else, each neighbor $v_j$ of
$u_i$ other than the root $r$ receives an additional load
$p_{u_iv_j} L_{u_i}^{(1)}$, where conditional on the degree
$N(u_i)+1$ of $u_i$, the collection $\{p_{u_iv_j}, j=1, \ldots
,N({u_i})\}$ forms an independent exchangeable random probability
mass function identical in distribution to the one above. The
probability that one of the neighbors of $u_j$ is also neighbor of
$r$ is much smaller than the probabilities of interest and can be
ignored due to the locally tree-like nature of the random graph as
is the standard practice in such analysis.

The process is said to terminate at time $T+1$ if some node at a
distance $T$ from the root has failed and none of the nodes at
distance $T+1$ fail at time $T+1$. Else the cascade continues. Our
first result is a sufficient condition for the cascade to survive
forever with positive probability. For two reals, $a,b$ let $a
\vee b$ denote their maximum.
\section{Super-Critical Regime}

\begin{thm} Consider the ARC process on the infinite random graph $G$
as described above.
Then if
\begin{equation}
\sum_{k=1}^{\infty} \tw_k k \bbE\left[
\bar{\sfF}(c(1-p(k)))\right]
> 1, \label{e1}
\end{equation}
then the cascade survives indefinitely with positive
probability,
%
where $\bar{\sfF} = 1 - \sfF$ and $p(k)$ is as defined in
(\ref{eq:distp}).
\label{t1}
\end{thm}
\begin{proof}
The idea of the proof is to couple the ARC process with a
super-critical Galton-Watson (GW) process as follows. Instead of
transferring the actual load at a failed node to its
neighbors/children, we will transfer an amount equal to the
capacity $c$ which is less than the load at the failed node. To be
specific, the GW process starts at time $0$ with a single
individual at the root. The $N(r)$ nodes $u_1, \ldots , u_r$
connected to the root are the possible children of the root in
generation $1$. If $N(r) = 0$, then the process terminates. If
$N(r) > 0$, then conditioned on $N(r)=k$, the random allocation
probabilities $\{p_{ru_i}, i=1, \ldots ,k\}$,
and the loads $\{L_{u_i}^{(0)}, i=1, \ldots, k\}$, the node $u_i$
comes alive at time $1$ in the GW process if the following
condition is satisfied,
\[ p_{ru_i} c + L_{u_i}^{(0)} \geq c. \]

Note that in comparison, the node $u_i$ fails in the ARC process if
\[ L_{u_i}^{(1)} = p_{ru_i} L_r^{(0)} + L_{u_i}^{(0)} \geq c. \]
This procedure is repeated at each level to construct the GW
process. In constructing the GW process, one can think of the net
load $L_u$ of the parent that fails being replaced by its capacity
$c$. Since at the time of failure, the load at node $u$,
$L_u^{(1)} > c$, if a node comes alive in the GW process then it
fails in the ARC process as well. Using the capacity and not the
actual load of the failed node to determine the number of children
in the GW process, not only gives a lower bound on the number of
failed nodes, it also ensures that number of offspring are
independent and identically distributed across generations $n \geq
1$ in the GW process.

Now we need to derive the condition for survival or
super-criticality of this GW process. For any node $u$ that fails
in any generation larger than one, let $N(u)+1$ be the degree of
that node. $N(u)$ has the size-biased distribution $\{\tw_k, k
\geq 0\}$. If $N(u) > 0$, then conditional on $N(u)=k$, let $v_1,
\ldots , v_{k}$, be the neighbors of $u$ in the next generation
and let $p_{uv_1}, \ldots , p_{uv_k}$ be the random allocation.
The probability that $v_k$ comes alive in the GW process is given
by
\begin{eqnarray}
q_k & = & P[p_{uv_i} c + L_{v_i}^{(0)} \geq c | N(u)=k] \nonumber \\
 & = & \bbE\left[\bar{\sfF}( c(1-p(k))\right] ,
\label{pe1}
\end{eqnarray}
where $\bar{\sfF} = 1 - \sfF$
and $p(k)$ is as defined in (\ref{eq:distp}).
Thus, the expected number of offspring of $u$ in the GW process
given $N(u) = k$ equals $kq_k$. Since $N(u)$ has distribution $\{
\tw_k \}$,  the GW process is super-critical as long as the mean
number of offspring is greater than $1$ which is true if condition
(\ref{e1}) holds. If the GW process is super-critical, then the
(XXX which) process survives indefinitely with positive probability (see
Athreya and Ney, 1972, pp. 7). As noted above, if the GW process
survives, then so does the ARC process.
\end{proof}

The following corollary is immediate.
\begin{corollary}
Suppose that the load allocation is uniform, that is, conditional
on the degree at a failed node (other than the root) being $N+1,$
each of the $N$ non-failed neighbor receives an equal fraction
$p=\f{1}{N}$ of the load of its failed parent, provided $N > 0$.
Then the condition (\ref{e1}) for the ARC process to survive
indefinitely with positive probability reduces to
\begin{equation}\label{eq:cor2}
 \sum_{k=1}^{\infty} \tw_k k \bar{\sfF}\left( \f{c(k-1)}{k}
\right) > 1. 
\end{equation}
%
If in addition, the degree of each node is fixed, that is  $N
\equiv d > 1$ is a constant, then  \eqref{eq:cor2} reduces to
\begin{equation}\label{eq:cor21}
d \left( 1 - \sfF \left( \f{c(d - 1)}{d} \right) \right) > 1.
\end{equation}
\end{corollary}
%
The next two results are straightforward consequences of the
coupling from below of the ARC process with a super-critical  GW
process as described in proof of Theorem ~\ref{t1}. In order to
state the results we need some notation.

Suppose that the load allocation is uniform, that is, $p(k) \equiv
1/k$. Given the random variable $N$ having the size-biased degree
distribution $\{ \tw_k, k \geq 0 \}$, let $M=0$ if $N=0$ and given
$N=j > 0$, let $M$ be a binomial random variable with parameters
$j$ and $q_j$, where
\begin{equation}\label{e1a}
q_j =  \bar{\sfF}\left(  \f{c(j-1)}{j} \right).
\end{equation}
%
Note that the $q_j$ defined in (\ref{pe1}) reduces to the one
defined in (\ref{e1a}) when the load allocation is uniform. Define
the probability mass function
\[ m_k = P[M = k] = \sum_{j=k}^{\infty} \binom{j}{k} q_j^k (1-q_j)^{j-k}
\tw_j , \]
for $k=0,1, \ldots$.
Let $\Phi(s) = \sum_{k=1}^{\infty} s^k m_k$, be the probability
generating function of $\{m_k, k \geq 0\},$ and let $\eta$ be the
smallest non-negative root of the equation $\Phi(s) = s$. With
uniform allocation, the offspring distribution of the coupled
branching process for the first and subsequent generations is same
as that of $M$. Note that the left hand side expression in
(\ref{eq:cor2}) is $\bbE[M]$. And $\ga = \bbE[M] > 1$ is the condition
for the branching process to be super-critical. Let $W_n$ be the
number of nodes that fail in the ARC process at time $n$, that is, these are nodes at
distance $n$ from the root and fail.
%
\begin{prop} Under condition (\ref{eq:cor2}), the ARC process survives
indefinitely with probability exceeding $1-\eta$. If in addition
we have $\sum_k k \log k \; m_k < \infty,$ then on a set of
probability at least $1 - \eta$, we have
\[ \liminf_{n \to \infty} \f{W_n}{\ga^n} > 0.\]
\label{p1}
\end{prop}
\begin{proof} The results follow easily
from the above coupling and the behavior of a super-critical
branching process (see Athreya and Ney, 1972, pp. 7, 30).
\end{proof}

The implication of this result is that when the cascade survives,
it grows at an exponential rate. Thus, the cascade spreads rapidly
if it does not die out quickly.

We now specialize the results for large finite graphs and examine the consequence of
the branching process coupling.
Let $G_n$ be a graph with $n$ vertices and degree distribution
$\{w_k, k \geq 0 \}$.
Suppose that the load from a failed parent is distributed equally
among all its non-failed neighbors, i.e. $p(k) \equiv 1/k$.
The number of first generation children of the coupled branching
process has distribution
\[ \ze_k =  \sum_{j=k}^{\infty} \binom{j}{k} q_j^k (1-q_j)^{j-k}
w_j , \qquad k=0,1, \ldots,\]
where $q_j$ is as defined in (\ref{e1a}).
%
%
Let $\Psi = \sum_{k=1}^{\infty} s^k \ze_k$ be the probability
generating function of $\{\ze_k, k \geq 0 \}$, and let $\eta$ be
as defined above. The following result is now a consequence of
Theorem 3.1.3 (\cite{durrett2007random}) and the branching process
coupling described in the proof of Theorem~\ref{t1}.
\begin{thm} A sufficient condition for the existence of a giant
component in $G_n$ is that (\ref{eq:cor2}) holds. If this condition
holds, then the fraction of vertices in the giant component
exceeds $1 - \Psi(\eta)$ asymptotically. \label{t1a}
\end{thm}
\begin{rem} If $\bbP(L > c) > 0$, then for a graph with a large
number of nodes, approximately a fraction $f = P[L>c]$ of nodes
will be in a failed state to begin with. In the sub-critical
regime, the cascades starting from these nodes will form small
islands of failed nodes, where each island is of finite size. 
In the super-critical regime, however, multiple initial failed will lead
to the formation of a unique giant component of failed nodes.
\end{rem}

\section{Sub-Critical Regime}

We now give a sufficient condition for the cascade to last only a
finite number of generations in the infinite random graph $G$
almost surely. Before we proceed, we need to describe a result of
Biggins \cite{Biggins1977} on branching random walks (BRW).

A BRW is a process that starts with a single individual labelled
$r$ located at $x \in \bbR$. An individual labelled $u$ born at
location $y \in \bbR$ lives for unit time at the end of which
gives birth to offspring located according to the point process $y
+ Z_u$, where $Z_u$ is an independent copy of a point process $Z$.
Let
\begin{eqnarray}
\delta(\th) & = & \bbE\left[ \int_{-\infty}^{\infty} e^{-\th t} dZ(t)\right],
\nonumber \\
 & = & \bbE\left[ \sum_u \exp(- \th z_u)\right], \label{BRW1}
\end{eqnarray}
where $\{ z_u \}$ are the points of $Z$. Define the function
\begin{equation}
\gamma(0) = \inf \{ \delta(\th): \th \geq 0 \}. \label{BRW2}
\end{equation}
The following result is a particular case of Theorem 2 of
\cite{Biggins1977}.

\begin{thm} \label{bigt2} Let $Z^{(n)}(0)$ be the number of individuals of the BRW
located in the interval $(-\infty, 0]$. If $\gamma(0) < 1$, then,
almost surely, $Z^{(n)}(0) = 0$ for all but finitely many $n$.
\end{thm}

Note that the result is independent of the location of the initial
individual. Let
\begin{equation}
h := \inf_{\theta \geq 0} \left\{ \bbE \left[ N e^{\th \left( L -
(1-p)c \right)}  \right] \right\}, \label{e2}
\end{equation}
where the random variable inside the expectation is taken to be
zero if $N=0$ and given $N=k > 0$, $p(k)$ is as defined in
(\ref{eq:distp}).
%
\begin{thm} Consider the ARC model on the graph $G$ as described
above. If $h < 1$, then the cascade starting at a node $r$ with
any load $\ell > c$ will terminate in finite time with probability
$1$.
%
%
\label{t2}
\end{thm}

\begin{proof}
We will dominate the ARC process with a BRW process starting with
a single individual located at $\ell$ to mimic the failing of the
root with load $L_r^{(0)} = \ell > c$.
%
%
To motivate the construction of the coupled BRW process, suppose
that $(u,v)$ is an edge in the random graph $G$ and the ARC
cascade upon reaching node $u$ at time $t$ results in the load at
$u$ increasing to $L_u^{t} \geq c$. This will cause the node $u$
to fail and result in the load at node $v$ at time $t+1$
increasing to
\[ L_v^{(t+1)} = p_{uv} L_u^{(t)} + L_v^{(0)}. \]
Hence the difference in the resultant loads at nodes $v$ and $u$,
or the ``drift'' in the load satisfies
\[ L_v^{(t+1)} - L_u^{(t)} = - (1-p_{uv}) L_u^{(t)} + L_v^{(0)} \leq - (1 -
p_{uv}) c + L_v^{(0)}, \]
where the last inequality follows since node $u$ fails at time $t$
and thus $ L_u^{(1)} \geq c$. Hence
\begin{equation}
L_v^{(t+1(} = L_u^{(t)} + (L_v^{(t+1)} - L_u^{(t)}) \leq L_u^{(t)} -
(1 - p_{uv}) c + L_v^{(0)}.
 \label{pe2}
\end{equation}
The coupled BRW process is defined as follows. The process starts
with a particle labelled $r$ located at $L_r^{(0)} = \ell > c$, which is the
load at the time of failure of the root node. Suppose node $u$
fails at time $t$ in the ARC process. If the number of non-failed
neighbors $N(u)$, of $u$ equals zero, then no child is born in the
BRW process. Else given $N(u)=k$, then for each non-failed
neighbor $v_i$ of $u$, an individual with the same label $v_i$ is born
in the BRW process at time $t+1$. If the new load at node $v_i$ in
the ARC process is
\[L_{v_i}^{(t+1)} = p_{uv_i}(k) L_u^{(t)} + L_{v_i}^{(0)},\]
then the location of $v_i$ in the BRW process will
\[x_{v_i}^{(t+1)} = x_u^{(t)} - (1-p_{uv_i}(k))c + L_v^{(0)},\]
where $x_u^{(t)}$ is the location of the individual labelled $u$
in the BRW process. A barrier is kept at $c$ in the BRW process,
that is if a child is born at location $x < c$, then it is killed.
%

Refer to Fig. \ref{fig:BRWdominance}, for an illustration of the
coupling argument. In Fig. \ref{fig:BRWdominance}, suppose at time 
$t$, a node labelled $u$ fails in the ARC process, then a particle
labelled $u$ is born in the BRW process located at $x^{(t)}_u \geq L_u^{(t)}$.
Now, since node $u$ has failed, a fraction $p_{uv}$ of the load at
node $u$ is transferred to node $v$, and if node $v$ fails as a
consequence, then a child of the particle at $x^{(t)}_u$ is born in the BRW process
at location $x^{(t+1)}_v$, where the displacement $x^{(t+1)}_v -
x^{(t)}_u = L_v^{(0)} - (1 - p_{uv}) c$. Thus, from \eqref{pe2},
if node $v$ fails at time $t+1$ in the ARC process, that is, $L_v^{(t+1)} \geq c$,
then a corresponding particle with label $v$ is born in the BRW
process and has location $x_v^{(t+1)} \geq L_v^{(1)}$. Thus, the
ARC process terminates if at any time there are no individuals in
the BRW.
%
\begin{center}
\begin{figure*}
\includegraphics[height=2in]{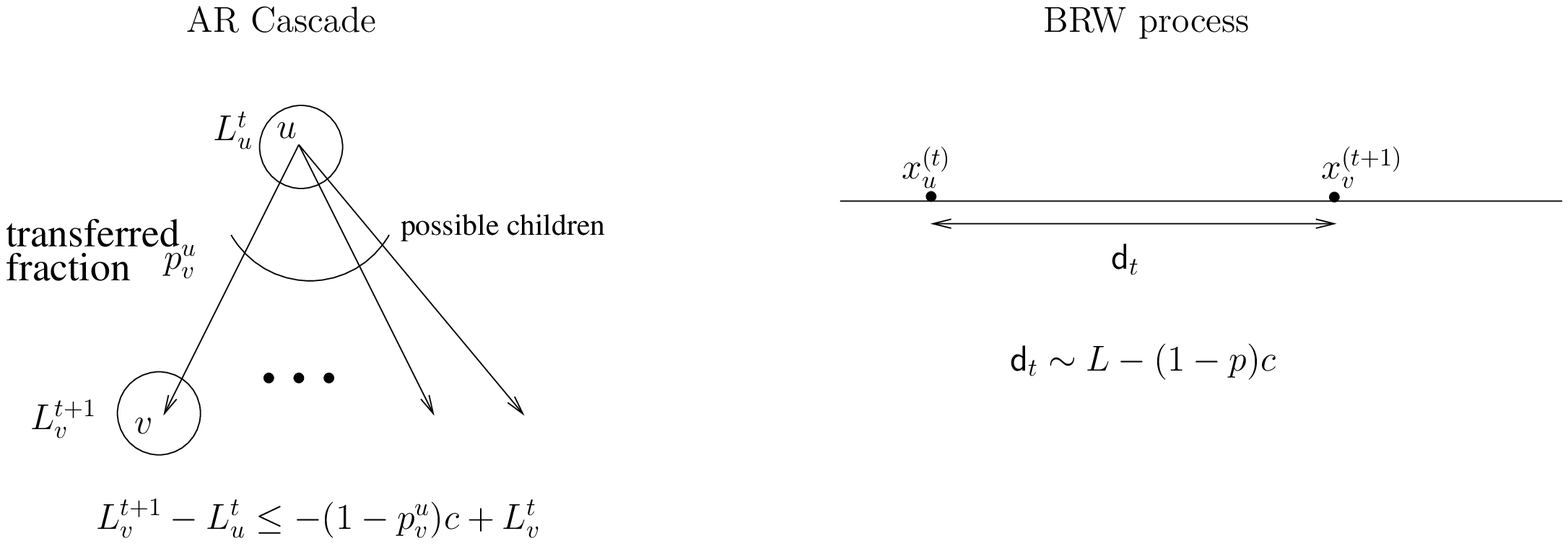}
\caption{Coupling the ARC model with BRW with a barrier.}
\label{fig:BRWdominance}
\end{figure*}
\end{center}
Let the coupled BRW process defined above be denoted by $X$. To
apply Theorem~\ref{bigt2}, we consider another BRW $Y$ coupled to
$X$ as follows. The $X, Y$ processes start off with one
individuals located at $\ell, -\ell$ respectively. 
Starting with the initial ancestors, if a child is born to a parent located at $x$ in the BRW $X$ such that the location of the child is $x+d$, then a child is born to a corresponding parent located at 
$y = -x$
in the process $Y$ and the child is located at $y-d$.
However, in the
process $Y$, we do not kill individuals born at locations to the left of $c$, that is, there is no
barrier. Thus in each generation, the individuals in $Y$ produce
offspring whose numbers are distributed as $N$ with displacements
distributed as $-(L-(1-p)c)$. 
Thus, the number of
individuals in $Y$ in any generation is at least as large as in
the process $X$. Moreover, note that if the BRW $Y$ drifts to the right,
then the $X$ process drifts to the left. 

Shift the origin to $-c$, so that the particle located at $x$ in process $X$ is at $x-c$ and 
particle located at $y$ in $Y$ process is at $-(y-c)$. Let $X_n([a, b]) (Y_n([a, b]))$ be
the number of individuals of process $X(Y)$ born at time $n$ and
are located in the interval $[a,b]$. Then the BRW $X$ terminates
if for some $n$, $X_n([2c, \infty)) = 0$ and this happens if
$Y_n((-\infty,0] = 0$. The location of the children of a node located at $0$ in the process $Y$ is given by $Z = \{
-(L_i - (1-p_i)c), i = 1, 2, \ldots ,N \}$ and $N$ has
distribution $\{\tw_k \}$. Thus for the process $Y$, we have
\[ \delta(\th) = \bbE \left[ N e^{\th \left( L -
(1-p)c \right)}  \right], \]
and hence $\gamma(0)=h$. Hence from Theorem~\ref{bigt2}, we have that
for $h < 1$, with probability one, the BRW $Y$ terminates  in
finite time. Hence the BRW $X$ terminates in finite time and
consequently, so does the ARC process.

\end{proof}

\section{Simulations}
In this section, to better understand the cascade dynamics we present some experimental results using Monte Carlo simulations, and compare our derived lower and upper bounds with the experimental thresholds at which cascade occurs. 
We consider both a random tree with a given degree distribution and a deterministic tree with fixed degree for each node. We let the tree to be of $N=10$ levels, and count the empirical measure of how often the cascade reaches any leaf at level $N$ as a measure of the number of failed nodes (probability of cascade) for large $N$ as a function of the capacity $c$. 

In Figs. \ref{fig:c+r2}, \ref{fig:c+r3}, and \ref{fig:c+r4}, we plot the cascade probability for a deterministic tree with degree $2$, $3$, and $4$ together with a random tree with Poisson distributed degree with mean $2$, $3$, and $4$, respectively. The load distribution is assumed to be a truncated Exponential distribution between $0.2$ and $2$, respectively. The vertical lines on the left and right denote our derived lower and upper bounds (similar color) for each of the plots. We notice that the derived lower and upper bounds are fairly tight and give us a good estimate of the true threshold value. Also, as we increase the degree value, the derived bounds become tighter. 

For a fixed tree (fixed number of children), the super-critical condition for cascade 
\eqref{eq:cor21} can be rewritten as :
$\sfF((d-1)c/d)<(d-1)/d$,
where $d$ is the number of child nodes,
$c$ is the capacity and 
$\sfF$ is the cumulative load distribution.
If we consider the load distribution to be uniform  between $[a \ b]$,
and let $\frac{d-1}{d} = \beta$, then the super-critical condition for cascade is 
$\frac{(\beta c - a )}{(b - a)} < \beta$. 
Recall that $b < c$ (since we have assumed that $\sfF$ is supported on $(0,c)$). Then if $a=0$,  condition $\frac{(\beta c - a )}{(b - a)} < \beta$ is never satisfied, 
so we pick $a$ greater than zero for all simulations. 
The same argument holds for the exponential load distribution case, and we simulate for truncated exponential load distribution between $[a,b]$ with $a>0$.

In Figs. \ref{fig:fexp}, \ref{fig:fpl}, and \ref{fig:fu}, we plot the cascade probability for deterministic tree with degree $2$, $3$, and $4$ for three different load distributions, exponential, power-law, and uniform, respectively. Once again, the vertical lines on the left and right denote our derived lower and upper bounds for cascade threshold for each of the plots. As before, we notice that the derived lower and upper bounds are fairly tight. Moreover, the bounds are tighter for exponential and power-law load distribution in comparison to uniform load distribution.

In Fig. \ref{fig:pla}, we plot the cascade probability for random tree with power-law degree distribution where degree $d$ is distributed as $P(d=k) = Ak^{-\alpha}$ for uniform load distribution between $[0.2 \ 2]$ as a function of parameter $\alpha$. For reasonable simulation complexity, we renormalize the degree distribution to have at most $10$ children, i.e. $k\le 10$. We see that as the power-law exponent $\alpha$ increases, the derived bounds become loose. 
Finally, in Fig. \ref{fig:Pa}, we plot the cascade probability for random tree with Poisson degree distribution for uniform load distribution between $[a, \ 2]$ as a function of parameter $a >0$.  We see that as $a$ increases, i.e., as the load distribution becomes more concentrated, the derived bounds become loose.

\begin{center}
\begin{figure*}
\includegraphics[height=3in]{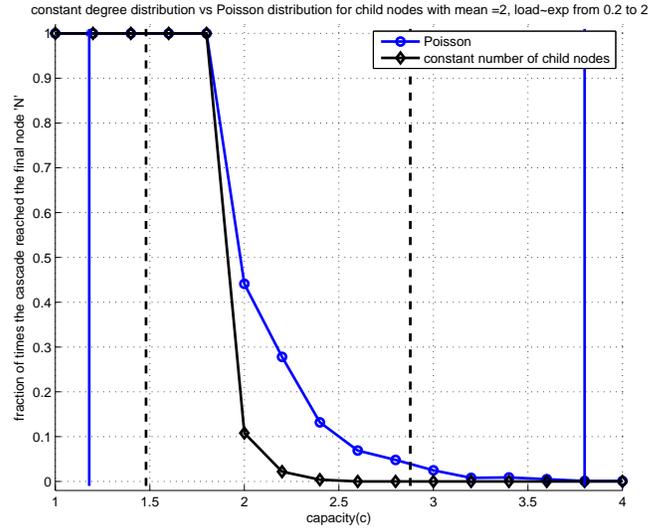}
\caption{Comparison of cascade probability with fixed and random tree with degree $2$.}
\label{fig:c+r2}
\end{figure*}
\end{center}

\begin{center}
\begin{figure*}
\includegraphics[height=3in]{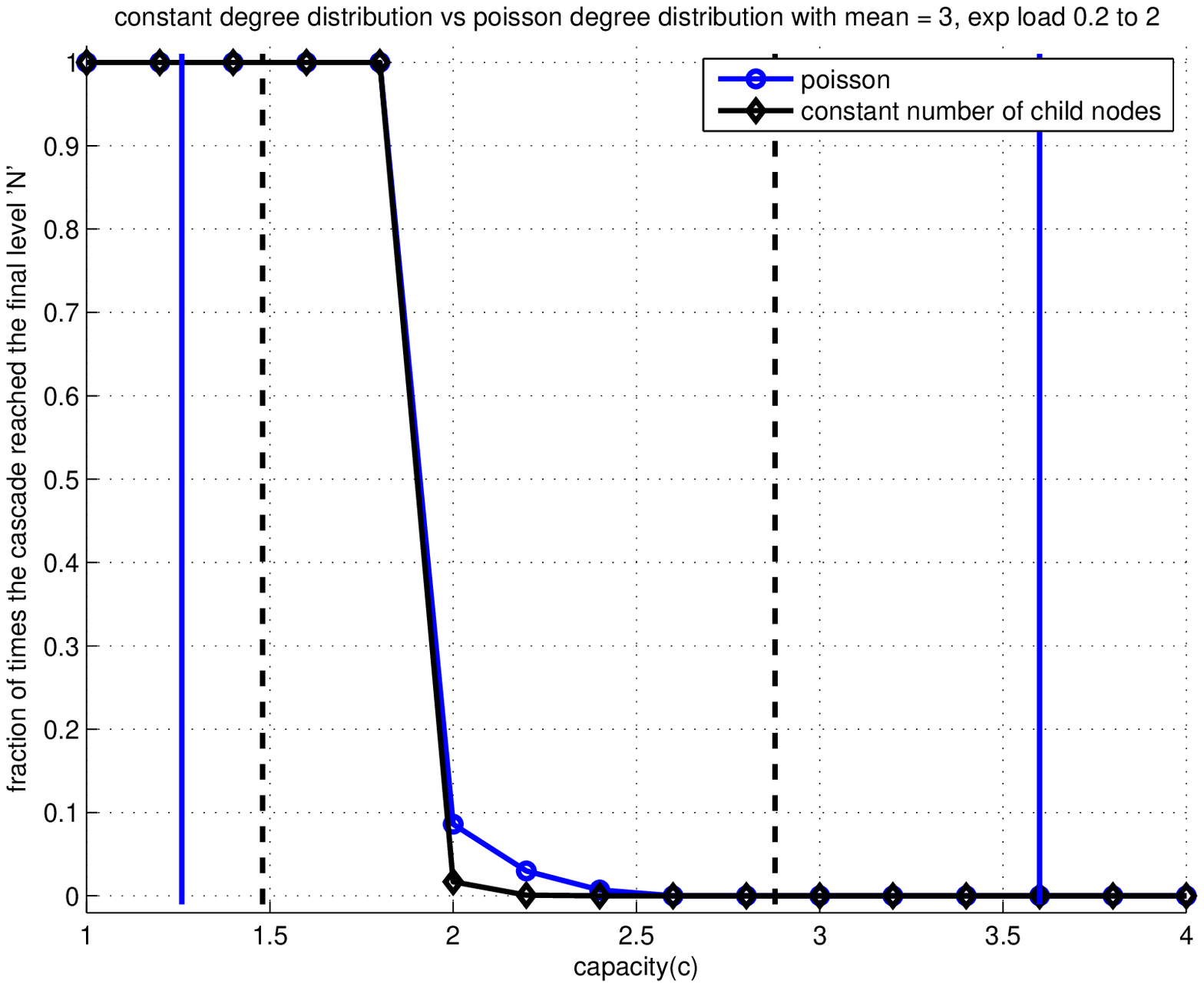}
\caption{Comparison of cascade probability with fixed and random tree with degree $3$.}
\label{fig:c+r3}
\end{figure*}
\end{center}

\begin{center}
\begin{figure*}
\includegraphics[height=3in]{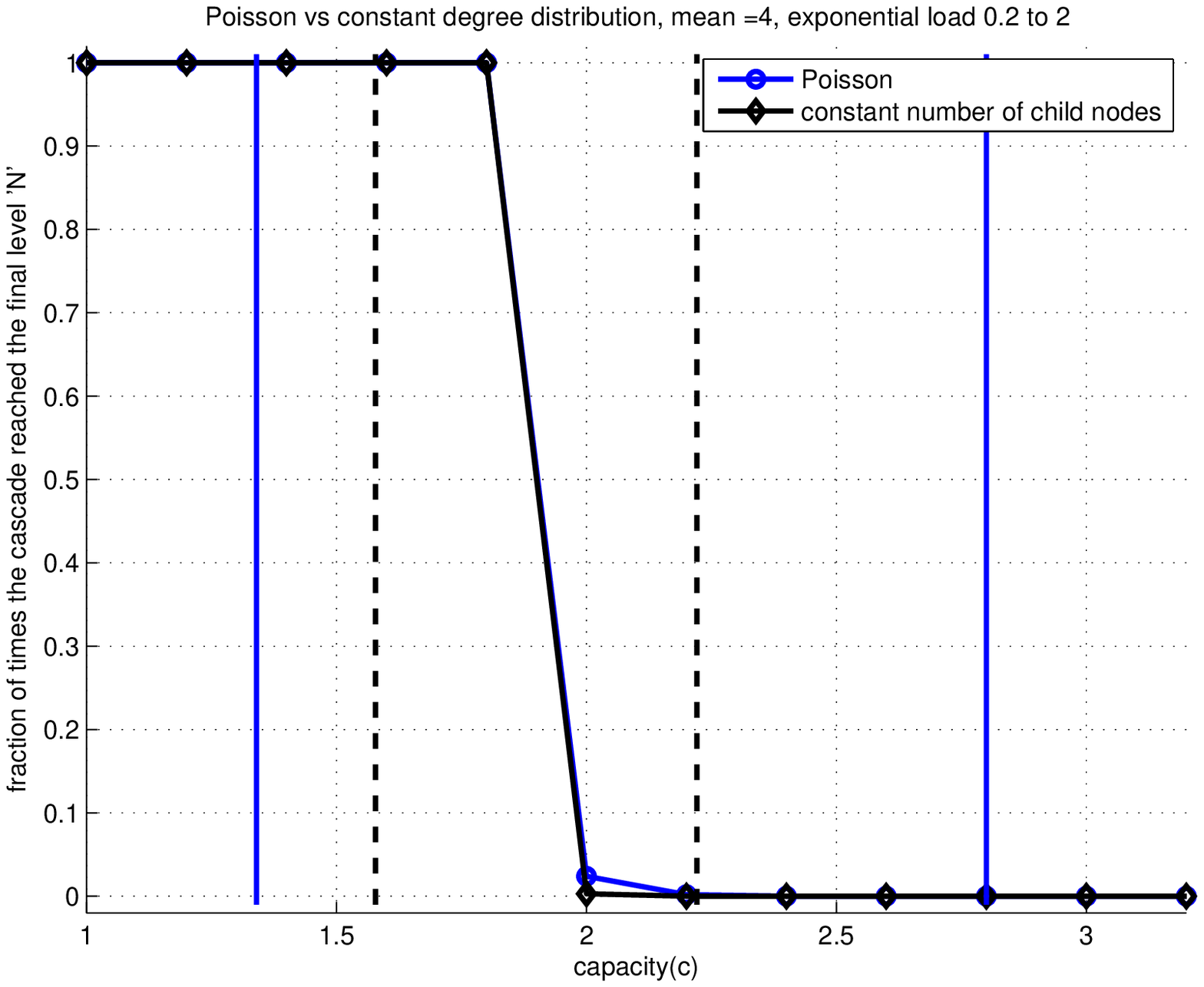}
\caption{Comparison of cascade probability with fixed and random tree with degree $4$.}
\label{fig:c+r4}
\end{figure*}
\end{center}

\begin{center}
\begin{figure*}
\includegraphics[height=3in]{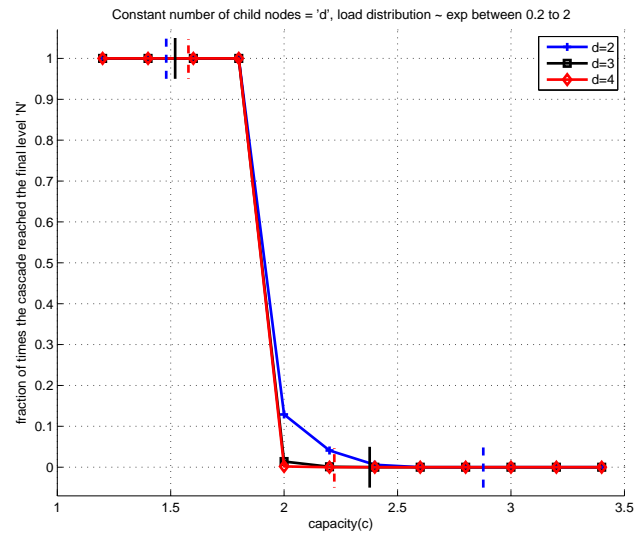}
\caption{Comparison of cascade probability with fixed tree for different degrees with exponential load distribution.}
\label{fig:fexp}
\end{figure*}
\end{center}

\begin{center}
\begin{figure*}
\includegraphics[height=3in]{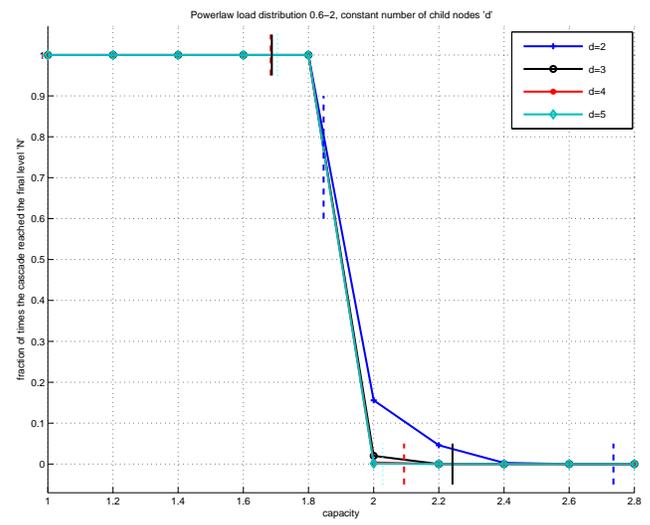}
\caption{Comparison of cascade probability with fixed tree  for different  degrees with power-law load distribution.}
\label{fig:fpl}
\end{figure*}
\end{center}

\begin{center}
\begin{figure*}
\includegraphics[height=3in]{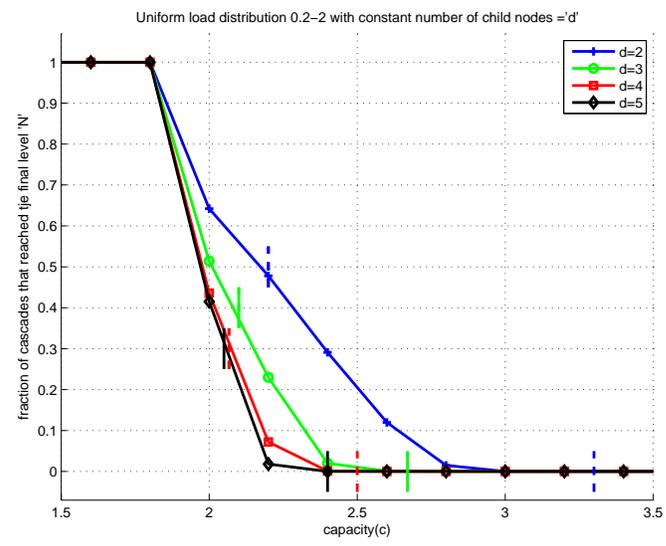}
\caption{Comparison of cascade probability with fixed  tree  for different  degrees with uniform load distribution.}
\label{fig:fu}
\end{figure*}
\end{center}

\begin{center}
\begin{figure*}
\includegraphics[height=3in]{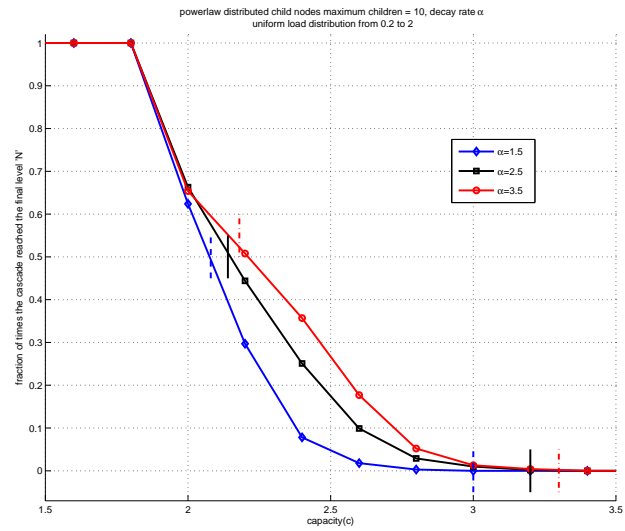}
\caption{Comparison of cascade probability with random tree with power-law degree distribution and uniform load distribution.}
\label{fig:pla}
\end{figure*}
\end{center}

\begin{center}
\begin{figure*}
\includegraphics[height=3in]{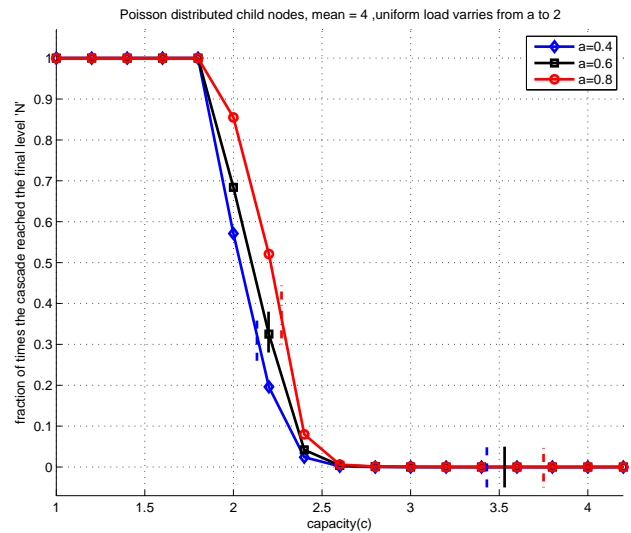}
\caption{Comparison of cascade probability with random tree with Poisson degree distribution and uniform load distribution.}
\label{fig:Pa}
\end{figure*}
\end{center}
%
%
%
%
%
%
%
%

\bibliographystyle{IEEEtran}
\bibliography{IEEEabrv,Research}

%
%
%
%
%
%
%
%
%
%

\end{document}